\documentclass[runningheads]{llncs}
\usepackage{hyperref}
\usepackage{breakurl}

\usepackage{graphicx}
\usepackage{pstricks}
\usepackage{amsmath}
\usepackage{amssymb}
\usepackage{mathrsfs}
\usepackage{tikz}
\usepackage{float}

\def\text#1{\textrm{#1}}

\def\precond#1{{\vphantom{#1}}^\bullet #1}
\def\postcond#1{{#1}^\bullet}

\def\production#1{\stackrel{#1}{\longrightarrow}}

\newfont{\fsc}{eusm10 scaled 1100}      
\def\powermultiset#1{\bbbn^{#1}}
\def\implies{\Rightarrow}

\def\mathrlap{\mathpalette\mathrlapinternal}
\def\mathrlapinternal#1#2{%
  \rlap{$\mathsurround=0pt#1{#2}$}}
\def\mathllap{\mathpalette\mathllapinternal}
\def\mathllapinternal#1#2{%
  \llap{$\mathsurround=0pt#1{#2}$}}

\def\epsilon{\varepsilon}

\def\defitem#1{\emph{#1}}

\def\AA{\text{\it AA}}

\let\origexists\exists
\let\orignexists\nexists
\let\origforall\forall
\def\quantorSpace{}
\def\exists#1.{\quantorSpace\origexists\def\quantorSpace{\,}#1.\onespace}
\def\nexists#1.{\quantorSpace\orignexists\def\quantorSpace{\,}#1.\onespace}
\def\forall#1.{\quantorSpace\origforall\def\quantorSpace{\,}#1.\onespace}
\def\onespace#1{\let\argument=#1\ifx\onespace#1\else~\fi\argument}

\let\origmin\min
\def\min{\mathord{\origmin}}
\let\origmax\max
\def\max{\mathord{\origmax}}

\def\quireunderscore{_}
\def\quire#1{%
  \def\tmp{#1}%
  \ifx\tmp\quireunderscore%
    \def\tmp{\quireindexed_}
  \else%
    \def\tmp{\mathscr{Q}#1}
  \fi\tmp}
\def\quireindexed_#1{\mathscr{Q}_{\text{#1}}}

\newtheorem{observation}{Observation}
\def\observationname{Observation}

\newcommand{\refdf}[1]{\definitionname~\ref{df-#1}}
\newcommand{\refpr}[1]{\propositionname~\ref{pr-#1}}
\newcommand{\refthm}[1]{\theoremname~\ref{thm-#1}}
\newcommand{\refcor}[1]{\corollaryname~\ref{cor-#1}}
\newcommand{\reflem}[1]{\lemmaname~\ref{lem-#1}}
\newcommand{\reffig}[1]{\figurename~\ref{fig-#1}}

\newcommand{\refobs}[1]{\observationname~\ref{obs-#1}}
\newcommand{\refsec}[1]{Section~\ref{sec-#1}}

\makeatletter
\def\goesto{\@transition\rightarrowfill}
\def\Goesto{\@transition\Rightarrowfill}
\def\ngoesto{\@transition\nrightarrowfill}
\def\nGoesto{\@transition\nRightarrowfill}
\def\@transition#1{\@ifnextchar[{\@@transition{#1}}{\@@transition{#1}[]}}
\newbox\@transbox
\newbox\@arrowbox
\def\rightarrowfill{$\m@th\mathord-\mkern-6mu%
  \cleaders\hbox{$\mkern-2mu\mathord-\mkern-2mu$}\hfill
  \mkern-6mu\mathord\rightarrow$}
\def\Rightarrowfill{$\m@th\mathord=\mkern-6mu%
  \cleaders\hbox{$\mkern-2mu\mathord=\mkern-2mu$}\hfill
  \mkern-6mu\mathord\Rightarrow$}
\def\@@transition#1[#2]%
   {\setbox\@transbox\hbox
       {\vrule height 1.5ex depth 1ex width 0ex\hskip0.25em$\scriptstyle#2$\hskip0.25em}
   \ifdim\wd\@transbox<1.5em
      \setbox\@transbox\hbox to 1.5em{\hfil\box\@transbox\hfil}\fi
   \setbox\@arrowbox\hbox to \wd\@transbox{#1}
   \ht\@arrowbox\z@\dp\@arrowbox\z@
   \setbox\@transbox\hbox{$\mathop{\box\@arrowbox}\limits^{\box\@transbox}$}
   \ht\@transbox\z@\dp\@transbox\z@
   \mathrel{\box\@transbox}}
\makeatother

\def\alignedcaption[#1&#2]{\mbox{\scriptsize $\mathllap{#1{}}\mathrlap{#2}$}}

\def\ie{i.e.\ }

\def\restrictedto{\mathop\upharpoonright}

\def\lastandname{\ and}
\newcommand{\dcup}{\stackrel{\mbox{\huge .}}{\cup}}   
\newcommand{\plat}[1]{\raisebox{0pt}[0pt][0pt]{#1}}   
\newcommand{\inp}{\mathbin\in}                        
\def\idx#1#2#3#4#5{
  \def\argone{#1}
  \def\argtwo{#2}
  \def\argthree{#3}
  \def\argfour{#4}
  \def\argfive{#5}
  \def\testprime{'}
  \def\testdprime{''}
  \def\testtprime{'''}
  {\vphantom{\argthree}}_%
    {\vphantom{\argfour}\argone}^%
    {\vphantom{\argfive}{\argtwo}}%
  \argthree_%
    {\vphantom{\argone}\argfour}%
    \ifx\argfive\testprime\argfive\else%
    \ifx\argfive\testdprime\argfive\else%
    \ifx\argfive\testtprime\argfive\else%
    ^{\vphantom{\argtwo}\argfive}\fi\fi\fi%
}

\newcommand{\monus}{\mathrel{\raisebox{-0pt}[0pt][0pt]{$
                      \stackrel{\raisebox{-5pt}[0pt][0pt]{\huge$\cdot$}}
                               {\raisebox{0pt}[0pt][0pt]{$-$}}$}}}

\def\swap{\mbox{swap}}
\def\swapeq{\equiv_1}
\def\swapeqi{\equiv_{1}}

\def\boldbracketa{%
  \psline[linewidth=0.1pt]{-}(0.15,-0.05)(0.05,-0.05)%
  \psline[linewidth=1.5pt]{-}(0.05,-0.05)(0.05,0.25)%
  \psline[linewidth=0.1pt]{-}(0.05,0.25)(0.15,0.25)%
  \phantom{[\,}%
}
\def\boldbracketb{%
  \psline[linewidth=0.1pt]{-}(0.0,-0.05)(0.1,-0.05)%
  \psline[linewidth=1.5pt]{-}(0.1,-0.05)(0.1,0.25)%
  \psline[linewidth=0.1pt]{-}(0.1,0.25)(0.0,0.25)%
  \phantom{]\,}%
}
\def\BD#1{\mathop{\boldbracketa#1\boldbracketb}}
\def\BDinf#1{{\boldbracketa#1\boldbracketb_{\psscalebox{0.8}{\scriptscriptstyle\infty}}}}
\def\BDify#1{{\it BD}(#1)}
\newenvironment{itemise}{\begin{list}{$\bullet$}{\leftmargin 12pt \labelwidth\leftmargin\advance\labelwidth-\labelsep \topsep 4pt \itemsep 2pt \parsep 2pt}}{\end{list}}
\newenvironment{itemisei}{\begin{list}{$-$}{\leftmargin 12pt \labelwidth\leftmargin\advance\labelwidth-\labelsep \topsep 4pt \itemsep 2pt \parsep 2pt}}{\end{list}}
\def\justempty{}
\newenvironment{define}[1]{\begin{definition}\rm\def\arga{#1}\ifx\justempty\arga~\else#1\fi\\\vspace{-6ex}\\\mbox{~}\begin{itemise}}{\end{itemise}\end{definition}}

\def\lastandname{~{\small\&}}           
\DeclareFontFamily{T1}{la}{}
\DeclareFontShape{T1}{la}{m}{n}{<->s*[0.8571428571]la14}{}

\def\processfont#1{\text{\fsc #1}}
\def\NN{\processfont{N}}
\def\SS{\processfont{S}}
\def\TT{\processfont{T}}
\def\FF{\processfont{F}}
\def\MM{\processfont{M}}
\def\PP{P}
\def\QQ{Q}

\newcommand{\GR}{{\rm GR}(N)}
\newcommand{\fGR}{{\rm GR}_{\it fin}(N)}



\makeatletter
\headerps@out{%
[%
  /BBox[0 0 1 1]%
  /_objdef{pdfRedHighlight} %
/BP pdfmark %
gsave 1 0 0 setrgbcolor 0.7 .setopacityalpha /Multiply .setblendmode 0.1 setlinewidth 0 0 moveto 1 0 lineto stroke grestore
[/EP pdfmark%
[%
  /BBox[0 0 1 1]%
  /_objdef{pdfGreenHighlight} %
/BP pdfmark %
gsave 0 1 0 setrgbcolor 0.7 .setopacityalpha /Multiply .setblendmode 0.1 setlinewidth 0 0 moveto 1 0 lineto stroke grestore
[/EP pdfmark%
[%
  /BBox[0 0 1 1]%
  /_objdef{pdfCyanHighlight} %
/BP pdfmark %
gsave 0 1 1 setrgbcolor 0.7 .setopacityalpha /Multiply .setblendmode 0.1 setlinewidth 0 0 moveto 1 0 lineto stroke grestore
[/EP pdfmark%
}
\def\@linkborderhighlight{pdfRedHighlight}
\def\@citeborderhighlight{pdfGreenHighlight}
\def\@urlborderhighlight{pdfCyanHighlight}
\define@key{PDF}{AP}{\pdf@addtoks{#1}{AP}}
\def\hyper@linkend{%
  \literalps@out{\strip@pt@and@otherjunk\baselineskip\space H.L}%
  \edef\Hy@tempcolor{%
    \csname @\hyper@currentlinktype borderhighlight\endcsname
  }%
  \pdfmark{%
    pdfmark=/ANN,%
    linktype=link,%
    Subtype=/Link,%
    PDFAFlags=4,%
    Dest=\hyper@currentanchor,%
    AP={%
      <</N {\Hy@tempcolor}>>%
    },%
    Border=0 0 0,%
    Raw=H.B%
  }%
  \Hy@endcolorlink
  \ifHy@breaklinks
  \else
    \Hy@LinkMath
    \Hy@SaveSpaceFactor
    \egroup
    \Hy@RestoreSpaceFactor
  \fi
}%
\def\hyper@link#1#2#3{%
  \Hy@VerboseLinkStart{#1}{#2}%
  \edef\Hy@tempcolor{\csname @#1borderhighlight\endcsname}%
  \begingroup
    \protected@edef\Hy@testname{#2}%
    \ifx\Hy@testname\@empty
      \Hy@Warning{%
        Empty destination name,\MessageBreak
        using `\Hy@undefinedname'%
      }%
      \let\Hy@testname\Hy@undefinedname
    \fi
    \pdfmark[{#3}]{%
      linktype={#1},%
      Border=0 0 0,%
      pdfmark=/ANN,%
      Subtype=/Link,%
      PDFAFlags=4,%
      AP={%
        <</N {\Hy@tempcolor}>>%
      },%
      Dest=\Hy@testname
    }%
  \endgroup
}
\def\hyper@linkurl#1#2{%
  \begingroup
    \Hy@pstringdef\Hy@pstringURI{#2}%
    \hyper@chars
    \leavevmode
    \pdfmark[{#1}]{%
      pdfmark=/ANN,%
      linktype=url,%
      Border=0 0 0,%
      Action={<<%
        /Subtype/URI%
        /URI(\Hy@pstringURI)%
        \ifHy@href@ismap
          /IsMap true%
        \fi
      >>},%
      Subtype=/Link,%
      AP={%
        <</N {\@urlborderhighlight}>>%
      },%
      PDFAFlags=4%
    }%
  \endgroup
}
\makeatother

\usepackage{pstricks}
\usepackage{pst-node}
\usepackage{graphicx}

\newcounter{netimage}
\def\p#1:#2;{\cnode #1{0.3}{n\thenetimage-#2}}
\def\P#1:#2;{\p #1:#2;\pscircle*#1{0.1}}
\def\q#1:#2:#3;{\p #1:#2;\rput#1{\rput[l](0.45,0){\large\it #3}}}
\def\Q#1:#2:#3;{\P #1:#2;\rput#1{\rput[l](0.45,0){\large\it #3}}}
\def\qq#1:#2:#3;{\p #1:#2;\rput#1{\rput[t](0,-0.5){\large\it #3}}}
\def\ql#1:#2:#3;{\p #1:#2;\rput#1{\rput[r](-0.45,0){\large\it #3}}}
\def\qt#1:#2:#3;{\p #1:#2;\rput#1{\rput[b](0,0.45){\large\it #3}}}
\def\Qt#1:#2:#3;{\P #1:#2;\rput#1{\rput[b](0,0.45){\large\it #3}}}
\def\Ql#1:#2:#3;{\P #1:#2;\rput#1{\rput[r](-0.45,0){\large\it #3}}}
\def\qx#1:#2:#3:#4;{\p #1:#2;\rput#1{\rput#4{\large\it #3}}}
\def\QXX#1:#2:#3:#4:#5;{\p #1:#2;\rput#1{\rput#4{\large\it #3}}\pscircle*#5{0.1}}
\def\s#1:#2:#3;{\p #1:#2;\rput#1{\rput(-0.03,0){\large\it #3}}}
\def\t#1:#2:#3;{\rput#1{\rnode{n\thenetimage-#2}{\psframebox{%
  \vbox to 0.6cm{\vfil\hbox to 0.6cm{\hfil\Large\it #3\hfil}\vfil}}}}}
\def\u#1:#2:#3:#4;{\rput#1{\rnode{n\thenetimage-#2}{\psframebox{%
  \vbox to 0.6cm{\vfil\hbox to 0.6cm{\hfil\Large\it #3\hfil}\vfil}}}}%
  \rput#1{\rput[l](0.6,0){\large\it #4}}}
\def\ut#1:#2:#3:#4;{\rput#1{\rnode{n\thenetimage-#2}{\psframebox{%
  \vbox to 0.6cm{\vfil\hbox to 0.6cm{\hfil\Large\it #3\hfil}\vfil}}}}%
  \rput#1{\rput[b](0,0.6){\large\it #4}}}
\def\ul#1:#2:#3:#4;{\rput#1{\rnode{n\thenetimage-#2}{\psframebox{%
  \vbox to 0.6cm{\vfil\hbox to 0.6cm{\hfil\Large\it #3\hfil}\vfil}}}}%
  \rput#1{\rput[r](-0.6,0){\large\it #4}}}
\def\a#1->#2;{\ncline{->}{n\thenetimage-#1}{n\thenetimage-#2}}
\def\A#1->#2;{\ncarc[arcangle=16]{->}{n\thenetimage-#1}{n\thenetimage-#2}}
\def\AA#1->#2;{\ncarc[arcangle=32]{->}{n\thenetimage-#1}{n\thenetimage-#2}}
\def\B#1->#2;{\ncarc[arcangle=-8]{->}{n\thenetimage-#1}{n\thenetimage-#2}}
\def\avlinearc{0.2}
\def\av#1[#2]-#3->[#4]#5;{
  \SpecialCoor
  \psline[linearc=\avlinearc]{->}([angle=#2]n\thenetimage-#1)#3([angle=#4]n\thenetimage-#5)
}

\makeatletter
\long\def\petrinet(#1)#2\end{\psscalebox{0.7}{\pspicture(#1)\stepcounter{netimage}#2\endpspicture}\end}

\makeatother

\psset{arrowsize=5pt}

\newcommand{\AND}{and}

\titlerunning{Abstract Processes and Conflicts in~Place/Transition Systems}
\title{Abstract Processes and Conflicts in~Place/Transition Systems%
\texorpdfstring{\thanks{This work was
partially supported by the DFG (German Research Foundation).}}{}}
\authorrunning{R.J. van Glabbeek, U. Goltz \& J.-W. Schicke-Uffmann}

\author{
  Rob van Glabbeek \inst{1,2} \and
  Ursula Goltz \inst{3} \and
  Jens-Wolfhard Schicke-Uffmann \inst{3}
}
\institute{
  Data61, CSIRO, Sydney, Australia
  \and
  School of Comp.\ Sc.\ and Engineering, Univ.\ of New South Wales,
  Sydney, Australia
  \and
  Institute for Programming and Reactive Systems, TU Braunschweig, Germany
}

\begin{document}

\maketitle
\setcounter{footnote}{0}

\begin{abstract}
For one-safe Petri nets or condition/event-systems, a \defitem{process} as
defined by Carl Adam Petri provides a notion of a run of a system where causal
dependencies are reflected in terms of a partial order.
Goltz and Reisig have generalised this concept for nets where places carry
multiple tokens, by distinguishing tokens according to their causal history.
However, this so-called \defitem{individual token interpretation} is often
considered too detailed.
Here we identify a subclass of Petri nets, called \defitem{structural conflict nets},
where no interplay between conflict and concurrency due to token multiplicity occurs.
For this subclass, we define abstract processes as equivalence classes of
Goltz-Reisig processes.  We justify this approach by showing that
there is a largest abstract process if and only if the underlying net is
conflict-free with respect to a canonical notion of conflict.
\end{abstract}

\section{Introduction}\label{sec-intro}
\noindent
In this paper we address a well-known problem in Petri net theory,
namely how to generalise Petri's concept of non-sequential processes
to nets where places may carry multiple tokens. We propose and justify
a solution for a subclass of Petri nets, called \defitem{structural conflict nets}.

One of the most interesting features of Petri nets is that they allow
the explicit representation of causal dependencies between action
occurrences when modelling reactive systems.
Petri defined \defitem{condition/event systems}, where ---
amongst other restrictions --- places (there called conditions) may
carry at most one token. For this class of nets, he proposed what is
now the classical notion of a \defitem{process}, given as a mapping
from an \defitem{occurrence net} (acyclic net with unbranched places)
to the original net \cite{petri77nonsequential,genrich80dictionary}.
A process models a run of the represented system, obtained by choosing
one of the alternatives in case of conflict. It records all
occurrences of the transitions and places visited during such a run,
together with the causal dependencies between them, which are given by
the flow relation of the net.

However, the most frequently used class of Petri nets are nets where places
may carry arbitrarily many tokens, or a certain maximal number of tokens when
adding place capacities. This type of nets is often called
\defitem{place/transition systems} (\mbox{P\hspace{-1pt}/T} systems). Here
tokens are usually assumed to be indistinguishable entities, for example
representing a number of available resources in a system. Unfortunately, it is
not straightforward to generalise the notion of process, as defined by Petri
for condition/event systems, to \mbox{P\hspace{-1pt}/T} systems. In fact, it
has now for more than 30 years been a well-known problem in Petri net theory
how to formalise an appropriate causality-based concept of process or run for
general \mbox{P\hspace{-1pt}/T} systems. In the following we give an
introduction to the problem and a short overview on existing approaches.

As a first approach, Goltz {\AND} Reisig generalised Petri's notion of process to
general \mbox{P\hspace{-1pt}/T} systems \cite{goltz83nonsequential}. We call
this notion of a process \defitem{GR-process}. It is based on a
canonical unfolding of a P\hspace{-1pt}/T system into a condition/event system,
representing places that may carry several tokens by a corresponding number
of conditions (see \cite{goltz87representations}).  \reffig{unsafe} shows a
\mbox{P\hspace{-1pt}/T} system with two of its GR-processes.

\begin{figure}
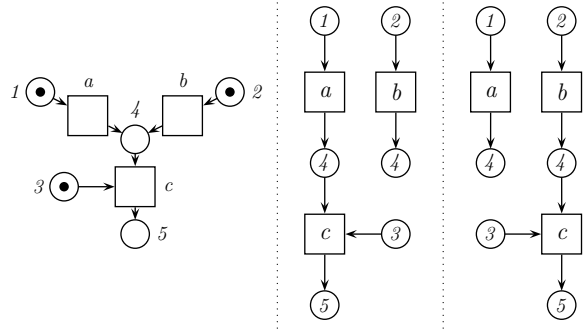

  \begin{center}
    \psscalebox{0.9}{
    \begin{petrinet}(12,7)
      \Ql(0.5,5):p:1;
      \Q(4.5,5):q:2;
      \ut(1.5,4.5):a::a;
      \ut(3.5,4.5):b::b;
      \qt(2.5,4):r:4;
      \Ql(1,3):cs:3;
      \u(2.5,3):c::c;
      \q(2.5,2):s:5;

      \a p->a; \a q->b; \a a->r; \a b->r;
      \a r->c; \a cs->c; \a c->s;

      \psline[linestyle=dotted](5.5,0)(5.5,7)

      \s(6.5,6.5):p2:1;
      \s(8.0,6.5):q2:2;
      \t(6.5,5):a2:a;
      \t(8.0,5):b2:b;
      \s(6.5,3.5):r2:4;
      \s(8.0,3.5):r2p:4;
      \t(6.5,2):c2:c;
      \s(8.0,2):cs2:3;
      \s(6.5,0.5):s2:5;

      \a p2->a2; \a q2->b2; \a a2->r2; \a b2->r2p; \a r2->c2; \a cs2->c2; \a c2->s2;

      \psline[linestyle=dotted](9.0,0)(9.0,7)

      \s(10.0,6.5):p3:1;
      \s(11.5,6.5):q3:2;
      \t(10.0,5):a3:a;
      \t(11.5,5):b3:b;
      \s(10.0,3.5):r3p:4;
      \s(11.5,3.5):r3:4;
      \t(11.5,2):c3:c;
      \s(10.0,2):cs3:3;
      \s(11.5,0.5):s3:5;

      \a p3->a3; \a q3->b3; \a a3->r3p; \a b3->r3; \a r3->c3; \a cs3->c3; \a c3->s3;

    \end{petrinet}
    }
  \end{center}
  \vspace{-2ex}
  \caption{A net $N$ with its two maximal GR-processes. The correspondence between elements of the net and their occurrences in the processes is indicated by labels.}
  \label{fig-unsafe}
\end{figure}

However, if one wishes to interpret \mbox{P\hspace{-1pt}/T} systems with a
causal semantics, there are alternative interpretations of what ``causal
semantics'' should actually mean. Goltz already argued that when abstracting
from the identity of multiple tokens residing in the same place, GR-processes
do not accurately reflect runs of nets, because if a Petri net is
conflict-free, in the sense that there are no choices to resolve, it should intuitively
have only one complete run, yet it may have multiple maximal GR-processes
\cite{goltz86howmany}.  This phenomenon occurs in \reffig{unsafe},
since the choice between alternative behaviours is here only due to the
possibility to choose between two tokens which can or even should be seen as
indistinguishable entities.  A similar argument is made, e.g., in
\cite{HKT95}.

At the heart of this issue is the question whether multiple tokens residing in
the same place should be seen as individual entities, so that a transition
consuming just one of them constitutes a choice, as in the interpretation
underlying GR-processes, or whether such tokens are
indistinguishable, so that taking one is equivalent to taking the other.  Van
Glabbeek {\AND} Plotkin call the former viewpoint the \defitem{individual token
interpretation} of P\hspace{-1pt}/T systems.  For an alternative
interpretation, they use the term \defitem{collective token interpretation}
\cite{glabbeek95configuration}.  A possible formalisation of these
interpretations occurs in \cite{glabbeek05individual}. In the following we
call process notions for \mbox{P\hspace{-1pt}/T~systems} which are adherent to
a collective token philosophy \defitem{abstract processes}.  Another option,
proposed by Vogler, regards tokens only as notation for a natural number
stored in each place; these numbers are incremented or decremented when firing
transitions, thereby introducing explicit causality between any transitions
removing tokens from the same place \cite{vogler91executions}.
\advance\textheight 1pt

\advance\textheight -1pt
Mazurkiewicz applies again a different approach in \cite{mazurkiewicz89multitree}.
He proposes \defitem{multitrees}, which record possible multisets of fired
transitions, and then takes confluent subsets of multitrees as abstract
processes of P\hspace{-1pt}/T systems. This approach does not explicitly
represent dependencies between transition occurrences and hence does not apply
to nets with self-loops, where such information may not always be retrieved.

Yet another approach has been proposed by Best {\AND} Devillers in
\cite{best87both}.  Here an equivalence relation is generated by a
transformation for changing causalities in GR-processes, called
\defitem{swapping}, that identifies GR-processes which differ only in
the choice which token was removed from a place. In this paper, we
adopt this approach and we show that it yields a fully satisfying
solution for a subclass of P\hspace{-1pt}/T systems. We call the
resulting notion of a more abstract process \defitem{BD-process}. In
the special case of one-safe \mbox{P\hspace{-1pt}/T} systems (where
places carry at most one token), or for condition/event systems, no
swapping is possible, and a BD-process is just an isomorphism class of
GR-processes.

Meseguer {\AND} Montanari formalise runs in a net $N$ as
morphisms in a category $\mathcal{T}(N)$ \cite{MM88}. In \cite{DMM89}
it has been established that these morphisms ``coincide with the
commutative processes defined by Best and Devillers'' (their
terminology for BD-processes). Likewise, Hoogers, Kleijn {\AND}
Thiagarajan\linebreak[4] \cite{HKT95} represent an abstract run of a net by a \defitem{trace},
thereby generalising the trace theory of Mazurkiewicz
\cite{mazurkiewicz95tracetheory}, and remark that ``it is
straightforward but laborious to set up a 1-1 correspondence between
our traces and the equivalence classes of finite processes generated
by the swap operation in [Best and Devillers, 1987]''.

As observed by Vogler \cite{vogler90swapping} (as a consequence of Corollary 5.6 therein),
it can be argued that BD-processes are not fully satisfying as
abstract processes for general \mbox{P\hspace{-1pt}/T} systems.
To illustrate this result, we recall in
\reffig{badswapping} an example due to Ochma\'nski
\cite{ochmanski89personal} --- see also
\cite{DMM89,glabbeek11ipl}. In the initial situation only two of the three
enabled transitions can fire, which constitutes a conflict.  However, the
equivalence obtained from the swapping transformation (formally defined in
\refsec{semantics}) identifies all possible maximal GR-processes and hence
yields only one complete abstract run of the system.  We are not aware of a
solution, i.e.\ any formalisation of the concept of a run of a net
that conforms to the collective token interpretation
and meets the requirement that for this net there is more than one complete run.

\begin{figure}[t]
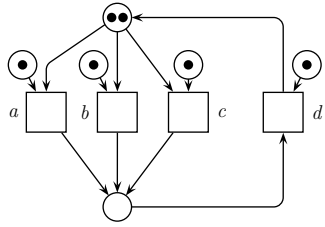

\vspace*{1em}
  \begin{center}
    \psscalebox{0.9}{
    \begin{petrinet}(12,5)
      \P(3.5,4):pa;
      \P(5,4):pb;
      \P(7,4):pc;
      \P(9.5,4):pd;

      \ul(4,3):a::a;
      \ul(5.5,3):b::b;
      \u(7,3):c::c;
      \u(9,3):d::d;

      \p(5.5,5):p;
      \p(5.5,1):q;

      \av p[210]-(4,4)->[90]a; \a pa->a; \a a->q;
      \a p->b; \a pb->b; \a b->q;
      \a p->c; \a pc->c; \a c->q;
      \av q[0]-(9,1)->[270]d; \a pd->d; \av d[90]-(9,5)->[0]p;

      \pscircle*(5.38,5){0.1}
      \pscircle*(5.62,5){0.1}
    \end{petrinet}
    }
  \end{center}
  \vspace{-5ex}
  \caption{A net with only a single process up to swapping equivalence.$\!$}
  \label{fig-badswapping}
\end{figure}

In \cite{glabbeek11ipl} and in the present paper, we continue the line of
research of \cite{MM88,DMM89,mazurkiewicz89multitree,HKT95} to formalise a
causality-based notion of an abstract process of a \mbox{P\hspace{-1pt}/T}
system that fits a collective token interpretation.  As remarked already in
\cite{goltz86howmany}, `what we need is some notion of an ``abstract
process''' and `a notion of maximality for abstract processes', such that `a
\mbox{P\hspace{-1pt}/T}-system is conflict-free iff it has exactly one maximal
abstract process starting at the initial marking'.  The example from
\reffig{badswapping} shows that BD-processes are in general not suited.
However, we show that BD-processes are completely adequate on
a subclass of P\hspace{-1pt}/T-systems --- proposed in \cite{glabbeek11ipl} --- where
conflict and concurrency are clearly separated.
We called these nets \defitem{structural conflict nets}.\linebreak[1]
Using the formalisation of conflict for
\mbox{P\hspace{-1pt}/T} systems from \cite{goltz86howmany}, we have shown in \cite{glabbeek11ipl}
that, for this subclass of P\hspace{-1pt}/T systems, we obtain
some finite BD-processes without a common extension
whenever the net contains a conflict.
The proof of this result is quite involved; it was achieved by
using an alternative characterisation of BD-processes via firing sequences
from \cite{best87both}.
As we point out in \refsec{IPLresult}, this result implies that a
structural conflict net with a largest BD-process must be conflict-free.

In this paper, we will show the reverse direction of this result, namely that
a structural conflict net has a largest BD-process if the net is conflict-free.
We then have established that there is a largest
abstract process in terms of BD-processes for structural conflict nets
\emph{if and only if} the net is conflict-free with respect to a canonical
notion of conflict.

We proceed by defining basic notions for \mbox{P\hspace{-1pt}/T} systems in
\refsec{basic}. In \refsec{semantics}, we define
GR-processes and BD-processes, as well as a natural partial order on BD-processes that gives
rise to the notion of a largest BD-process.
\refsec{conflict} recalls the concept of conflict in \mbox{P\hspace{-1pt}/T}
systems and defines structural conflict nets.\footnote{The material in
Sections \ref{sec-basic}, \ref{sec-GR} and \ref{sec-conflict} follows closely the
presentation in \cite{glabbeek11ipl}, but needs to be included to make the
paper self-contained.}
In \refsec{IPLresult} we recall a result from \cite{glabbeek11ipl} that implies that
a structural conflict net featuring any conflict cannot have a largest BD-process.
In \refsec{largestprocess} we prove the converse, 
that a conflict-free structural conflict net does have a largest BD-process.
Finally, \refsec{maximality} reformulates (and slightly strengthens) our result in the terminology of
\cite{glabbeek11ipl}, where we did not employ a partial order on BD-processes, and hence no canonical
notion of a largest BD-process. We show that a structural conflict net is conflict-free iff it has a
unique maximal GR-process up to swapping equivalence.

The results of this paper, together with a slightly extended overview on the
existing approaches on semantics of Petri nets, were previously announced in \cite{GGS11a},
with proofs in an accompanying technical report. Our current proofs are conceptually much simpler,
as they are carried out directly on BD-processes, rather than via the auxiliary concepts of BD-runs and
FS-runs. This became possible after turning swapping equivalence into a preorder on BD-processes,
simply by employing only one of the two symmetric clauses defining this relation. That idea stems
from Walter Vogler [personal communication, 20-11-2012], whom we gratefully acknowledge.

Additionally, we thank the referees for their very detailed and helpful comments.

\section{Place/transition systems}\label{sec-basic}

\noindent
We will employ the following notations for multisets.

\begin{define}{
  Let $X$ be a set.
}\label{df-multiset}
\item A {\em multiset} over $X$ is a function $A\!:X \rightarrow \bbbn$,
i.e.\ $A\in \powermultiset{X}\!\!$.
\item $x \in X$ is an \defitem{element of} $A$, notation $x \in A$, iff $A(x) > 0$.
\item For multisets $A$ and $B$ over $X$ we write $A \subseteq B$ iff
 \mbox{$A(x) \leq B(x)$} for all $x \inp X$;
\\ $A\cup B$ denotes the multiset over $X$ with $(A\cup B)(x):=\text{max}(A(x), B(x))$,
\\ $A\cap B$ denotes the multiset over $X$ with $(A\cap B)(x):=\text{min}(A(x), B(x))$,
\\ $A + B$ denotes the multiset over $X$ with $(A + B)(x):=A(x)+B(x)$,
\\ $A - B$ is given by
$(A - B)(x):=A(x)\monus B(x)=\mbox{max}(A(x)-B(x),0)$, and
for $k\inp\bbbn$ the multiset $k\cdot A$ is given by
$(k \cdot A)(x):=k\cdot A(x)$.
\item The function $\emptyset\!:X\rightarrow\bbbn$, given by
  $\emptyset(x):=0$ for all $x \inp X$, is the \emph{empty} multiset over $X$.
\item If $A$ is a multiset over $X$ and $Y\subseteq X$ then
  $A\restrictedto Y$ denotes the multiset over $Y$ defined by
  $(A\restrictedto Y)(x) := A(x)$ for all $x \inp Y$.
\item The cardinality $|A|$ of a multiset $A$ over $X$ is given by
  $|A| := \sum_{x\in X}A(x)$.
\item A multiset $A$ over $X$ is \emph{finite}
  iff $|A|<\infty$, i.e.,
  iff the set $\{x \mid x \inp A\}$ is finite.
\item A function $\pi: X \rightarrow Y$ extends to multisets $A \in \powermultiset{X}$ by
  \plat{$\displaystyle\pi(A)(y) = \!\sum_{y = \pi(x)}\!A(x)$}. In this paper, this sum
  will always turn out to be finite.
\end{define}
Two multisets $A\!:X \rightarrow \bbbn$ and
$B\!:Y\rightarrow \bbbn$
are \emph{extensionally equivalent} iff
$A\restrictedto (X\cap Y) = B\restrictedto (X\cap Y)$,
$A\restrictedto (X\setminus Y) = \emptyset$, and
$B \restrictedto (Y\setminus X) = \emptyset$.
In this paper we often do not distinguish extensionally equivalent
multisets. This enables us, for instance, to use $A \cup B$ even
when $A$ and $B$ have different underlying domains.
With $\{x,x,y\}$ we will denote the multiset over $\{x,y\}$ with
$A(x)\mathbin=2$ and $A(y)\mathbin=1$, rather than the set $\{x,y\}$ itself.
A multiset $A$ with $A(x) \leq 1$ for all $x$ is
identified with the set $\{x \mid A(x)=1\}$.

Below we define place/transition systems as net structures with an initial marking.
In the literature we find slight variations in the definition of \mbox{P\hspace{-1pt}/T}
systems concerning the requirements for pre- and postsets of places
and transitions. In our case, we do allow isolated places. For
transitions we allow empty postsets, but require at least one
preplace, thus avoiding problems with infinite self-concurrency.
Moreover, following \cite{best87both}, we restrict attention
to nets of \defitem{finite synchronisation}, meaning that each
transition has only finitely many pre- and postplaces.
Arc weights are included by defining the flow relation as a function to the natural numbers.
For succinctness, we will refer to our version of a \mbox{P\hspace{-1pt}/T} system as a \defitem{net}.

\begin{define}{}\label{df-nst}
\item[]
  A \defitem{net} is a tuple
  $N = (S, T, F, M_0)$ where
  \begin{itemise}
    \item $S$ and $T$ are disjoint sets (of \defitem{places} and \defitem{transitions}),
    \item $F: ((S \mathord\times T) \mathrel\cup (T \mathord\times S)) \rightarrow \bbbn$
      (the \defitem{flow relation} including \defitem{arc weights}), and
    \item $M_0 : S \rightarrow \bbbn$ (the \defitem{initial marking})
  \end{itemise}
  such that for all $t \inp T$ the set $\{s\mid F(s, t) > 0\}$ is
  finite and non-empty, and the set $\{s\mid F(t, s) > 0\}$ is finite.
\end{define}

\noindent
Graphically, nets are depicted by drawing the places as circles and
the transitions as boxes. For $x,y \inp S\cup T$ there are $F(x,y)$
arrows (\defitem{arcs}) from $x$ to $y$.\footnote{This is a presentational alternative for the
  common approach of having at most one arc from $x$ to $y$, labelled with the \emph{arcweight}
  $F(x,y) \in \bbbn$.}  When a net represents a
concurrent system, a global state of this system is given as a
\defitem{marking}, a multiset of places, depicted by placing $M(s)$
dots (\defitem{tokens}) in each place $s$.  The initial state is
$M_0$.

\begin{define}{
  Let $N\!=\!(S, T, F, M_0)$ be a net and $x\inp S\cup T$.
}\label{df-preset}
\item[]
The multisets $\precond{x},~\postcond{x}: S\cup T \rightarrow
\bbbn$ are given by $\precond{x}(y)=F(y,x)$ and
$\postcond{x}(y)=F(x,y)$ for all $y \inp S \cup T$.
If $x\in T$, the elements of $\precond{x}$ and $\postcond{x}$ are
called \emph{pre-} and \emph{postplaces} of $x$, respectively.
These functions extend to multisets
$X:S \cup T \rightarrow\bbbn$ as usual, by
$\precond{X} := \sum_{x \in S \cup T}X(x)\cdot\precond{x}$ and
$\postcond{X} := \sum_{x \in S \cup T}X(x)\cdot\postcond{x}$.
\end{define}
The system behaviour is defined by the possible moves between
markings $M$ and $M'$, which take place when a finite multiset $G$ of
transitions \defitem{fires}.  When firing a transition, tokens on
preplaces are consumed and tokens on postplaces are created, one for
every incoming or outgoing arc of $t$, respectively.  Obviously, a
transition can only fire if all necessary tokens are available in $M$
in the first place. \refdf{firing} formalises this notion of behaviour.

\begin{define}{
  Let $N \mathbin= (S, T, F, M_0)$ be a net,
  $G \in \bbbn^T\!$, $G$ non-empty and finite, and $M, M' \in \bbbn^S\!$.
}\label{df-firing}
\item[]
$G$ is a \defitem{step} from $M$ to $M'$,
written $M\production{G}_N M'$, iff
\begin{itemise}
  \item $^\bullet G \subseteq M$ ($G$ is \defitem{enabled}) and
  \item $M' = (M - \mbox{$^\bullet G$}) + G^\bullet$. 
\end{itemise}
We may leave out the subscript $N$ if clear from context.
For a word $\sigma = t_1t_2\ldots t_n \in T^*$
we write $M\production{\sigma} M'$ for\vspace{-5pt}
$$
\exists M_1, M_2, \ldots, M_{n-1}.
M\!\production{\{t_1\}}\! M_1\!\production{\{t_2\}}\! M_2 \cdots M_{n-1}\!\production{\{t_n\}}\! M'\!\!.
$$
When omitting $\sigma$ or $M'$ we always mean it to be existentially quantified.
\end{define}

\noindent
Note that steps are (finite) multisets, thus allowing self-concurrency.
Also note that $M\goesto[\{t,u\}]$ implies $M\goesto[tu]$ and $M\goesto[ut]$.

\section{Processes of place/transition systems}\label{sec-semantics}

\noindent
We now define two notions of a process of a net, modelling a run of
the represented system on two levels of abstraction.

\subsection{GR-processes}\label{sec-GR}

A (GR-)process is essentially a conflict-free, acyclic net together
with a mapping function to the original net. It can be obtained by
unwinding the original net, choosing one of the alternatives in case
of conflict.
The acyclic nature of the process gives rise to a notion of causality
for transition firings in the original net via the mapping function.
A conflict present in the original net is represented by the existence of
multiple processes, each representing one possible way to decide the conflict.

\begin{define}{}\label{df-process}
 \item[]
  A pair $\PP = (\NN, \pi)$ is a
  \defitem{(GR-)process} of a net $N = (S, T, F, M_0)$
  iff
  \begin{itemise}\itemsep 3pt
   \item $\NN = (\SS, \TT, \FF, \MM_0)$ is a net, satisfying
   \begin{itemisei}
    \item $\forall s \in \SS. |\precond{s}| \leq\! 1\! \geq |\postcond{s}|
    \wedge\, \MM_0(s) = \left\{\begin{array}{@{}l@{\quad}l@{}}1&\mbox{if $\precond{s}=\emptyset$}\\
                                   0&\mbox{otherwise,}\end{array}\right.$
    \item $\FF$ is acyclic, \ie
      $\forall x \inp \SS \cup \TT. (x, x) \mathbin{\not\in} \FF^+$,
      where $\FF^+$ is the transitive closure of $\{(x,y)\mid \FF(x,y)>0\}$,
    \item and $\{t \in \TT \mid (t,u)\in \FF^+\}$ is finite for all $u\in \TT$.
   \end{itemisei}
    \item $\pi:\SS \cup \TT \rightarrow S \cup T$ is a function with 
    $\pi(\SS) \subseteq S$ and $\pi(\TT) \subseteq T$, satisfying
   \begin{itemisei}
    \item $\pi(\MM_0) = M_0$, i.e.\ $M_0(s) = |\pi^{-1}(s) \cap \MM_0|$ for all $s\in S$, and
    \item $\forall t \in \TT, s \in S.
      F(s, \pi(t)) = |\pi^{-1}(s) \cap \precond{t}| \wedge
      F(\pi(t), s) = |\pi^{-1}(s) \cap \postcond{t}|$, i.e.\
      $\forall t \in \TT. \pi (^\bullet t)= {^\bullet \pi(t)} \wedge
      \pi (t^\bullet)= {\pi(t)^\bullet}$.
  \end{itemisei}
  \end{itemise}
  $P$ is called \defitem{finite} if $\TT$ is finite. The \defitem{end of} $P$
  is defined as $P^\circ = \{s \in \SS \mid \postcond{s} = \emptyset\}$.%
\end{define}

\noindent
Let $\GR$ (resp.~$\fGR$) denote the collection of (finite) GR-processes of~$N\!$.

A process is not required to represent a completed run of the original net.
It might just as well stop early. In those cases, some set of transitions can
be added to the process such that another (larger) process is obtained. This
corresponds to the system taking some more steps and gives rise to a natural
order between processes.

\begin{define}{
  Let $\PP = ((\SS, \TT, \FF, \MM_0), \pi)$ and $\PP' = ((\SS', \TT\,', \FF\,', \MM_0'), \pi')$ be
  two processes of the same net.
}\label{df-extension}
\item
  $\PP'$ is a \defitem{prefix} of $\PP$, notation $\PP'\leq \PP$, and 
  $\PP$ an \defitem{extension} of $\PP'$, iff 
    $\SS'\subseteq \SS$,
    $\TT\,'\subseteq \TT$,
    $\MM_0' = \MM_0$,
    $\FF\,'=\FF\restrictedto(\SS' \mathord\times \TT\,' \mathrel\cup \TT\,' \mathord\times \SS')$
    and $\pi'=\pi\restrictedto(\SS'\cup \TT\,')$.
\item
  A process of a net is said to be \defitem{maximal} if 
  it has no proper extension.
\end{define}

\noindent
The requirements above imply that if $\PP'\leq \PP$, $(x,y)\in
\FF^+$ and $y\in \SS' \cup \TT\,'$ then $x\in \SS' \cup \TT\,'$.
Conversely, any subset $\TT\,'\subseteq \TT$ satisfying
$(t,u)\in \FF^+ \wedge u\in \TT\,' \Rightarrow t\in \TT\,'$ uniquely determines a
prefix of $\PP$, denoted $\PP\mathbin\upharpoonright\TT\,'$.

In \cite{petri77nonsequential,genrich80dictionary,goltz83nonsequential} processes were
defined without requiring the third condition on $\NN$ from \refdf{process}.
Goltz {\AND} Reisig \cite{goltz83nonsequential} observed that certain
processes did not correspond with runs of systems, and proposed to
restrict the notion of a process to those that can be approximated by
finite processes \cite[end of Section~3]{goltz83nonsequential}.
This is the role of the third condition on $\NN$ in \refdf{process}; it
is equivalent to requiring that each transition occurs in a finite prefix.
In \cite{petri77nonsequential,genrich80dictionary,goltz83nonsequential}
only processes of finite nets were considered. For those processes,
the requirement of \emph{discreteness} proposed in \cite{goltz83nonsequential}
is equivalent with imposing the third condition on $\NN$ in \refdf{process}
\cite[Theorem 2.14]{goltz83nonsequential}.

Two processes $\PP \mathbin= (\NN, \pi)$ and $\PP' \mathbin= (\NN\,', \pi')$
are \defitem{isomorphic}, notation $\PP \cong \PP'$, iff there exists
an isomorphism $\phi$ from $\NN$ to $\NN\,'$ which respects the
process mapping, i.e.\ $\pi = \pi' \circ \phi$.
Here an \emph{isomorphism} $\phi$ between two nets $\NN=(\SS, \TT, \FF,
\MM_0)$ and $\NN\,'=(\SS', \TT\,', \FF\,', \MM'_0)$ is a
bijection between their places and transitions such that
$\MM'_0(\phi(s))=\MM_0(s)$ for all $s\in\SS$ and
$\FF\,'(\phi(x),\phi(y))=\FF(x,y)$ for all $x,y\in \SS\cup\TT$.

\subsection{BD-processes}

Next we formally introduce the swapping transformation and the resulting\linebreak[4]
equivalence notion on GR-processes from \cite{best87both}.

\begin{define}{
  Let $\PP = ((\SS, \TT, \FF, \MM_0), \pi)$ be a process and
  let $p, q \in \SS$ with $(p,q) \notin \FF^+\cup (\FF^+)^{-1}$ and
  $\pi(p) = \pi(q)$.
  }
\label{df-swap}
\item[]
  Then $\swap(\PP, p, q)$ is defined as $((\SS, \TT, \FF\,', \MM_0), \pi)$ with
  \begin{equation*}
    \FF\,'(x, y) = \begin{cases}
      \FF(q, y) & \text{ iff } x = p,\, y \in \TT\\
      \FF(p, y) & \text{ iff } x = q,\, y \in \TT\\
      \FF(x, y) & \text{ otherwise. }
    \end{cases}
  \end{equation*}
\end{define}

\noindent
  If $P$ is the first process depicted in \reffig{unsafe}, with $p$ and $q$ the two places
  that are mapped to place $4$ of the underlying net $N$, then $\swap(\PP, p, q)$ is the second
  process of \reffig{unsafe}. The transformation simply swaps the arcs leaving $p$ and $q$.

\begin{define}{}\label{df-swapeq}
\item
  Two processes $\PP$ and $\QQ$ of the same net are
  \defitem{one step swapping equivalent} ($\PP \swapeq \QQ$) iff
  $\swap(\PP, p, q)$ is isomorphic to $\QQ$ for some places $p$ and $q$.
\item
We write $\swapeq^*$ for the reflexive and transitive closure of $\swapeq$.
\pagebreak[3]
\end{define}
By taking $p\mathbin=q$ in \refdf{swapeq} one finds that $P \mathbin{\swapeq} P$ for any
non-empty process~$P\!$.
In \cite[Definition 7.8]{best87both} swapping equivalence---there denoted $\equiv_1^\infty$---is defined
in terms of \emph{reachable B-cuts}. Using \cite[Definition 3.14]{best87both}
this definition can be rephrased as follows:
\begin{define}{Let $N$ be a net, and $P,Q\inp \GR$.}\label{df-BD-swapping}
\item[]
Then $\PP \equiv_1^\infty \QQ$ iff
\begin{equation}\label{swapBD}
\forall P''\inp\fGR, P'' \leq P. \exists Q'\swapeq^* Q,~ Q''\leq
Q'. P'' \cong Q''
\end{equation}
and, vice versa,
$$\forall Q''\inp\fGR, Q'' \leq Q. \exists P'\swapeq^* P,~ P''\leq P'. P'' \cong Q''\;.$$
\end{define}
In \cite[Theorem 7.9]{best87both} (as well as below) it is shown that
$\equiv_1^\infty$ is an equivalence relation on GR-processes.
Trivially, $\swapeq^*$ is included in $\swapeq^\infty$.
\begin{define}{}
\item[]
We call a $\swapeq^\infty$-equivalence class of GR-processes
a \defitem{{\small BD-}process}, and write $\BDinf{P}$.
\end{define}
To support the idea that $\swapeq^\infty$ is a natural
equivalence relation on GR-processes, an alternative characterisation of
$\swapeq^\infty$ is presented in \refsec{characterisation} below.

In order to establish concepts of a maximal and of a largest BD-process, we turn
$\equiv_1^\infty$ into a preorder by focusing on only one of the two
clauses of \refdf{BD-swapping} (formulated differently).

\begin{define}{Let $N$ be a net, and $P,Q\in \GR$.}\label{df-BD-swapping-alt}
\item[]
Then $\PP \sqsubseteq_1^\infty \QQ$ iff\vspace{-1ex}
\begin{equation}\label{swapGGS}
\forall P''\inp\fGR, P'' \leq P. \exists P',Q'\in\fGR.
  P'' \leq P' \swapeqi^* Q' \leq Q.
\end{equation}
\vspace{-4ex}
\end{define}
We proceed to show that (\ref{swapGGS}) is equivalent to
(\ref{swapBD}) and that $\sqsubseteq_1^\infty$ is a preorder.

For $P,Q\inp\GR$ write $P\sim_S Q$ if there are places $p$ and $q$ such that
$\swap(\PP, p, q)=Q$ (or equivalently $\swap(Q,p,q)=P$).
Clearly, the operations of swapping two places in a process, and of bijectively
renaming all places and transitions, commute:
\begin{observation}\label{obs-swapiso}
$\exists P'\inp\GR. P\cong P' \sim_S Q
 ~\Leftrightarrow~
 \exists Q'\inp\GR. P\sim_S Q' \cong Q$.
\end{observation}
The same holds for the operations of taking a prefix and of bijectively
renaming all places and transitions:
\begin{observation}\label{obs-isoprefix}
$\exists P'\inp\GR. P\cong P' \leq Q
 ~\Leftrightarrow~
 \exists Q'\inp\GR. P\leq Q' \cong Q$.
\end{observation}
Moreover, instead of first swapping two places $p$ and $q$ in a process and then
extending the resulting process, we can just as well first extend and then swap:
\begin{observation}\label{obs-swapprefix}
$\exists P'\inp\GR. P\sim_S P' \leq Q
 ~\Rightarrow~
 \exists Q'\inp\GR. P\leq Q' \sim_S Q$.
\end{observation}
This implication can in general not be reversed, since it could be that out of two
swapped places $p$ and $q$ occurring in $Q$ and $Q'$ only one occurs in $P$.
\begin{lemma}\label{lem-swapprefix}\rm
If $P\leq Q' \sim_S Q$ for some $P\inp\fGR$ and $Q',Q\inp\GR$, then there are
$P',P''\inp\fGR$ with $P \leq P' \sim_S P'' \leq Q$.
\pagebreak[3]
\end{lemma}
\begin{proof}
Let $Q=\swap(Q',p,q)$ for certain places $p$ and $q$ in $Q'$. Take a finite
prefix $P'$ of $Q'$ that includes $P$ as well as $p$ and $q$. Then
$P\leq P'\leq Q'$. Let $P'':=\swap(P',p,q)$. Then $P' \sim_S P'' \leq Q$.
\qed
\end{proof}
$P\swapeq Q$ is defined by $\exists Q'. P\sim_S Q' \cong Q$.
Using that $\sim_S$ is reflexive on nonempty processes we have:
\begin{observation}\label{obs-iso-swap}
${\cong} \subseteq {\swapeq^*}$.
\end{observation}
So, using \refobs{isoprefix}, (\ref{swapBD}) can be reformulated as
$$\forall P''\inp\fGR, P'' \leq P. \exists Q'\swapeqi^* Q. P''\leq Q'.$$
From Observations~\ref{obs-isoprefix} and~\ref{obs-swapprefix} we obtain:
\begin{corollary}\label{cor-swapprefix}
$\exists P'\inp\GR. P\swapeqi^* P' \leq Q
 \Rightarrow
 \exists Q'\inp\GR. P\leq Q' \swapeqi^* Q$.
\qed
\end{corollary}
Likewise, from Observation~\ref{obs-isoprefix} and~\ref{obs-swapprefix} and Lemma~\ref{lem-swapprefix} we obtain:
\begin{corollary}\rm\label{cor-swapprefix-back}
If $P\leq Q' \swapeqi^* Q$ for some $P\inp\fGR$ and $Q',Q\inp\GR$, then there are
$P',P''\inp\fGR$ with $P \leq P' \swapeqi^* P'' \leq Q$.
\qed
\end{corollary}
Together, Corollaries~\ref{cor-swapprefix} and~\ref{cor-swapprefix-back} imply the equivalence of
(\ref{swapBD}) and (\ref{swapGGS}).
\refcor{swapprefix}, in combination with the transitivity of $\leq$ and
$\swapeqi^*$, implies the transitivity of $\sqsubseteq_1^\infty$.
Moreover, by definition $\sqsubseteq_1^\infty$ is reflexive.
\begin{corollary}\rm\label{cor-swapprefix-preorder}
$\sqsubseteq_1^\infty$ is a preorder on $\GR$. Hence $\equiv_1^\infty$
  is an equivalence relation.
\end{corollary}
It follows that $\sqsubseteq_1^\infty$ induces a partial order on
BD-processes, and thereby
concepts of a \emph{maximal} and a \emph{largest} BD-process.%
\footnote{A \emph{preorder} is a relation that is reflexive and transitive; it is an
  \emph{equivalence relation} if it moreover is symmetric, and a \emph{partial order} if it moreover
  is antisymmetric. Given a preorder $\sqsubseteq$, its \emph{kernel} is the equivalence relation
  $\equiv$ defined by $P \equiv Q$ iff $P \sqsubseteq Q \wedge Q \sqsubseteq P$. Moreover, the
  \emph{induced} partial order $\leq$ on the $\equiv$-equivalence classes is given by
  $[P] \leq [Q]$ iff $P \sqsubseteq Q$; it is easy to check that this is independent of the choices
  of representatives $P$ and $Q$ within the $\equiv$-equivalences classes $[P]$ and $[Q]$.

  Here $\equiv_1^\infty$ is the kernel of $\sqsubseteq_1^\infty$, and ``maximal'' or ``largest''
  refers to the induced partial order on BD-processes.
}

\subsection{An alternative characterisation of swapping equivalence}\label{sec-characterisation}

Let $P\in\GR$. The set $\BDify{P}$ of \emph{finite BD-approximations} of
$P$ is the smallest set of finite GR-processes that contains all
finite prefixes of $P$ and is closed under $\equiv_1$ and taking prefixes.
(By Observation~\ref{obs-iso-swap} it therefore is also closed under $\cong$.)
Thus, it is the smallest subset of $\GR$ satisfying
\begin{itemize}
\item if $P'\in\fGR$ and $P' \leq P$ then $P'\in \BDify{P}$,
\item if $P' \equiv_1 Q\in \BDify{P}$ then $P'\in\BDify{P}$,
\item if $P' \leq Q\in\BDify{P}$ then $P'\in \BDify{P}$.
\end{itemize}

\begin{proposition}\rm
$P \mathbin{\sqsubseteq_1^\infty} Q \Leftrightarrow \BDify{P} \mathbin{\subseteq} \BDify{Q}$.
So $P \mathbin{\swapeq^\infty} Q \Leftrightarrow \BDify{P} \mathbin= \BDify{Q}$.
\end{proposition}
\begin{proof}
$\BDify{P} = \{P''' \mathbin\in \fGR \mid \exists P'',P'\mathbin\in\fGR. P''' \leq P'' \swapeqi^* P' \leq P\}$
by \refcor{swapprefix}. Using this, the result follows directly from \refdf{BD-swapping-alt}.
\qed
\end{proof}

\section{Conflicts in place/transition systems}\label{sec-conflict}

We recall the notion of conflict introduced in
\cite{goltz86howmany}.
It formalises the notion of conflict alluded to in \cite[p.~23]{Reisig85a}.

\begin{define}{
  Let $N \mathbin= (S, T, F, M_0)$ be a net and $M \in
  \bbbn^S\!$.
}\label{df-semanticconflict}
\item
  A finite, non-empty multiset $G \in \bbbn^T$ is in
  \defitem{(semantic) conflict} in $M$ iff\vspace{2pt}\\
  $\neg M\goesto[G] ~~\wedge~~ \forall t \in G. M\goesto[G \restrictedto \{t\}]$.
\item
  $N$ is \defitem{(semantically) conflict-free} iff
  no finite, non-empty multiset $G \in \bbbn^T$ is in semantic conflict in any
  $M$ with $M_0 \goesto[] M$.
\item
  $N$ is \defitem{binary-conflict-\!-free} iff
  no multiset $G \in \bbbn^T$ with $|G| = 2$ is in semantic conflict in any
  $M$ with $M_0 \goesto[] M$.
\end{define}
Thus, $N$ is binary-conflict-\!-free iff whenever two different
transitions $t$ and $u$ are enabled at a reachable marking $M$, then
also the step $\{t,u\}$ is enabled at $M$.
The above concept of (semantic) conflict-freeness formalises the intuitive notion that
there are no choices to resolve.

\begin{trivlist}
\item[\hspace{\labelsep}\bf Remark:]
In a net such as displayed in \reffig{persistent},
the multiset $\{t,t\}$ is never enabled.
For this reason the multiset $\{t,t,u\}$ does not count as being in
conflict, even though it is never enabled. However, its subset
$\{t,u\}$ is in conflict.
\end{trivlist}

\begin{figure}[ht]
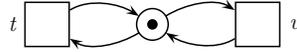

\vspace*{-5ex}
  \begin{center}
    \begin{petrinet}(6,2)
      \P(3,1):p;

      \ul(1,1):t::t;
      \u(5,1):u::u;

      \AA p->t; \AA t->p;
      \AA p->u; \AA u->p;
    \end{petrinet}
  \end{center}
  \vspace{-4ex}
  \caption{A net which is persistent but not binary-conflict--free}
  \label{fig-persistent}
\end{figure}

\noindent
A number of alternative concepts of conflict and conflict-freeness have been
contemplated in the Petri net community.

A Petri net $N$ is called \emph{persistent} \cite{KM69,LR78} if for every
marking $M$ with $M_0 \goesto[] M$ and every $t,u\in T$ with $t\neq u$, $M
\goesto[t]$ and $M\goesto[u]$, we have $M\goesto[tu]$; in other words, if any
transition $u$ that is enabled in a reachable marking will still be enabled after
firing any other transition $t$. Trivially, a net that is binary-conflict-\!-free
is also persistent. The net of \reffig{persistent}, on the other hand, is
persistent but not binary-conflict-\!-free.

\begin{figure}[ht]
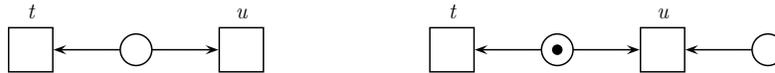

\vspace*{-3ex}
  \begin{center}
    \begin{petrinet}(16,2)
      \p(3,1):p;

      \ut(1,1):t::t;
      \ut(5,1):u::u;

      \a p->t;
      \a p->u;

      \P(11,1):p2;
      \p(15,1):q2;

      \ut(9,1):t2::t;
      \ut(13,1):u2::u;

      \a p2->t2;
      \a p2->u2;
      \a q2->u2;
    \end{petrinet}
  \end{center}
  \vspace{-4ex}
  \caption{Two nets with structural conflict, but no choices to resolve.}
  \label{fig-conflict-free}
\end{figure}

  A pair of different transitions in a net that share a preplace can be called a
  \emph{structural conflict}.
  As illustrated in \reffig{conflict-free}, the presence of a structural conflict
  does not imply that there are choices to resolve.
  A net that is free of structural conflicts is certainly conflict-free, but \reffig{conflict-free}
  shows that the reverse does not hold.

  A triple $(M,t,u)$ of a reachable marking $M$ and two different transitions $t$ and $u$
  with $M {\goesto[t]}$, $M {\goesto[u]}$ and $\precond{t} \cap \precond{u} \neq \emptyset$
  could be called a \emph{reachable structural conflict}. This constitutes a
  middle ground between semantic and structural conflict. The nets of \reffig{conflict-free} do not
  have a reachable structural conflict. However, the net of \reffig{reachable-conflict-free},
  although semantically conflict-free, does have a reachable structural conflict.
\begin{figure}[ht]
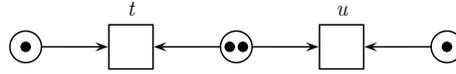

  \vspace*{-5ex}
  \begin{center}
    \begin{petrinet}(10,2)
      \P(1,1):p;
      \p(5,1):q;
      \pscircle*(5.125,1){0.1}
      \pscircle*(4.875,1){0.1}
      \P(9,1):r;

      \ut(3,1):t::t;
      \ut(7,1):u::u;

      \a p->t; \a q->t;
      \a q->u; \a r->u;
    \end{petrinet}
  \end{center}
  \vspace{-4ex}
  \caption{A net with a reachable structural conflict, but no choices to resolve.}
  \label{fig-reachable-conflict-free}
\end{figure}

Landweber and Robertson \cite{LR78} define a Petri net to be conflict-free
``if every place which is an input of more than one transition is on a self-loop with each such transition.''
This is an extension of the concept structural conflict-freeness that is closer to persistence.
It classifies the net of 
\reffig{persistent} as conflict-free and the nets of \reffig{conflict-free} as having conflicts.
Hence, this notion, just as persistence and structural conflict-freeness,
does not formalise the intuitive concept ``no choices to resolve''.

We proposed in \cite{glabbeek11ipl} a class of \mbox{P\hspace{-1pt}/T} systems
where (semantic) conflict-freeness coincides with the absence of reachable structural conflicts.
We called this class of nets \defitem{structural conflict nets}.
For a net to be a structural conflict net, we required that two
transitions sharing a preplace will never occur both in one step.

\begin{define}{
  Let $N \mathbin= (S, T, F, M_0)$ be a net.
}\label{df-structuralconflict}
\item[]
  $N$ is a \defitem{structural conflict net} iff
  $\forall t, u.
    (M_0 \goesto[]\;\goesto[\{t, u\}]) \implies
    \precond{t} \cap \precond{u} = \emptyset$.
\end{define}
Note that this excludes self-concurrency from the possible behaviours in a
structural conflict net: as in our setting every transition has at least one
preplace, $t = u$ implies $\precond{t} \cap \precond{u} \ne \emptyset$.  Also
note that in a structural conflict net a non-empty, finite multiset $G$ is in
conflict in a reachable marking $M$ iff $G$ is a set and two distinct transitions in $G$
are in conflict in $M$. Hence a structural conflict net is conflict-free if
and only if it is binary-conflict-\!-free.  Moreover, two transitions enabled
in $M$ are in (semantic) conflict iff they share a preplace.

Trivially, the class of structural conflict nets includes the class of safe nets, in
which no reachable marking assigns multiple tokens to the same place. It also includes the non-safe
net of \reffig{unsafe}, as well as the \emph{buffer synchronised systems of sequential machines} from
\cite{reisig82buffersync} and the \emph{locally sequential globally asynchronous nets (LSGA nets)} of
\cite{GGS13}, in which asynchronous communication is modelled by buffer-places between
sequential components that may collect arbitrarily many tokens.

\section{A structural conflict net having a largest BD-process is conflict-free}
\label{sec-IPLresult}

The result announced in this section---that each structural conflict net having a
$\sqsubseteq_1^\infty$-largest BD-process must be conflict-free---is in essence obtained in \cite{glabbeek11ipl}.
However, there we had not defined the order $\sqsubseteq_1^\infty$, and thus neither the
corresponding notion of a $\sqsubseteq_1^\infty$-largest BD-process. Instead we used different
terminology, and the work in this section merely consists of relating the terminology of
\cite{glabbeek11ipl} to the one of the present paper.

In \cite{glabbeek11ipl} a \emph{partial BD-run} of a net $N$ is defined as a $\swapeq^*$-equivalence class
of finite GR-processes of $N$.\footnote{It is easy to see that on finite GR-processes the relations
  $\swapeq^\infty$ and $\swapeq^*$ coincide. Hence a partial BD-run is the same as a finite
  BD-process, i.e., an equivalence class $\BDinf{P}$ with $P$ a finite GR-process.
  We do not use this fact further on.} Let $\BD{P}$ be the partial BD-run containing $P$.
The prefix/extension relation $\leq$ on $\fGR$ from \refdf{extension} is
lifted to partial BD-runs by \mbox{$\BD{P'} \leq \BD{P}$ iff $\PP' \swapeq^* \QQ' \leq \QQ \swapeq^* \PP$}
for some $\QQ',\QQ\in\fGR$. By \refcor{swapprefix}, $\leq$ is a partial order on partial BD-runs, and
$\BD{P'} \leq \BD{P}$ iff \mbox{$\PP' \leq \QQ \swapeq^* \PP$} for some $\QQ\in\fGR$.
Moreover, $\BD{P'} \leq \BD{P}$ iff $P' \sqsubseteq_1^\infty P$.

In \cite{glabbeek11ipl} a \emph{BD-run} of a net $N$ is defined as a prefix-closed and directed set of
partial BD-runs of $N$. Here we define the notion of a \emph{collapsed BD-run}, or BD$^\dagger$-run.

\begin{definition}\rm
A \emph{BD$^\dagger$-run} of a net $N$ is a subset $\mathcal{R}$ of $\fGR$ that is prefix-closed and
closed under $\swapeq^*$, and satisfies
\[P,Q\in\mathcal{R} \Rightarrow \exists P',Q'\in\mathcal{R}. P \leq P' \swapeq^* Q' \geq Q\;.\]
  \end{definition}
Note that a BD$^\dagger$-run is a set of finite GR-processes, whereas a BD-run is a set of sets of
finite GR-processes. We proceed to show that the two notions have the same information content.
For a BD-run $R$, let $R^\dagger:=\{P\inp\fGR \mid \BD{P} \inp R\}$.
Trivially, $R^\dagger$ is a BD$^\dagger$-run. Moreover, $R_1 \subseteq R_2$ implies \plat{$R_1^\dagger \subseteq R_2^\dagger$}.

Conversely, for $\mathcal{R}$ a BD$^\dagger$-run, let $\BD{\mathcal{R}} := \{\BD{P} \mid P \in \mathcal{R}\}$.
Trivially,  $\BD{\mathcal{R}}$ is a BD-run. Moreover, $\mathcal{R}_1 \subseteq \mathcal{R}_2$
implies $\BD{\mathcal{R}_1} \subseteq \BD{\mathcal{R}_2}$. Also note that
$\BD{R^\dagger} = R$ for any BD-run $R$, and $\BD{\mathcal{R}}^\dagger = \mathcal{R}$ for any
BD$^\dagger$-run $\mathcal{R}$. Thus we have a $\subseteq$-preserving bijective correspondence
between BD-runs and BD$^\dagger$-runs. It follows that a net has a unique maximal BD-run
iff it has a unique maximal BD$^\dagger$-run.\footnote{Exactly as in the proof of
  \reflem{towardsinfinity} in the next section it follows that each BD-run is a prefix of a maximal
  BD-run. Hence a unique maximal BD-run is the same as largest BD-run. The same applies to
  BD$^\dagger$-runs.}

\cite[Section 5]{glabbeek11ipl} defines the concept of an \emph{FS-run}---an FS run is a
certain set of sets of firing sequences---and establishes a $\subseteq$-preserving bijective
correspondence between FS-runs and BD-runs. It follows that a net has a unique maximal FS-run
iff it has a unique maximal BD-run.

The set of finite prefixes of a GR-process $P$ is directed: for $P_1,P_2\in\fGR$ with
$P_1\leq P$ and $P_2 \leq P$, there is a process $P_3\in\fGR$ with $P_1 \leq P_3 \leq P$ and
$P_2 \leq P_3$. Just take as transitions of $P_3$ the union of the transitions from $P$ that occur
in $P_1$ or $P_2$.

\begin{lemma}\rm\label{lem-largest}
  If a net has a $\sqsubseteq_1^\infty$-largest BD-process then it has a largest BD$^\dagger$-run.
\end{lemma}
\begin{proof}
  Let $\BDinf{P}$ be the $\sqsubseteq_1^\infty$-largest BD-process of a net $N$.
  We claim that the set of \emph{all} finite GR-processes of $N$ is a BD$^\dagger$-run.
  Clearly, it is then also the largest.

  Trivially, $\fGR$ is prefix-closed and closed under $\swapeq^*$.
  Now suppose $P_1,Q_1\linebreak[1]\in\fGR$. Since $P_1 \sqsubseteq_1^\infty P$ and $P_1 \leq P_1$ one
  has $P_1 \leq P_2 \swapeq^* P_3 \leq P$ for some $P_2,P_3\in\fGR$.
  Likewise $Q_1 \leq Q_2 \swapeq^* Q_3 \leq P$ for some $Q_2,Q_3\in\fGR$.
  Using that the set of prefixes of $P$ is directed, let $P_4\in\fGR$ be such that
  $P_3 \leq P_4 \leq P$ and $Q_3 \leq P_4$. Now \refcor{swapprefix} yields
  $P_1 \leq \swapeq^* P_4 \swapeq^* \geq Q_1$, which needed to be established.
\qed
\end{proof}

\begin{theorem}\rm\label{thm-conflictfree}
  Let $N$ be a structural conflict net.

  If $N$ has a $\sqsubseteq_1^\infty$-largest BD-process then $N$ is conflict-free.
\end{theorem}
\begin{proof}
\cite[Theorem 6]{glabbeek11ipl} says that if a structural conflict net $N$ has exactly one maximal
FS-run then $N$ is conflict-free.

Now suppose $N$ has a $\sqsubseteq_1^\infty$-largest BD-process. By \reflem{largest} it has a unique
maximal BD$^\dagger$-run. Hence it has a unique maximal BD-run, and a unique maximal FS-run.
It follows that $N$ is conflict-free.
\qed
\end{proof}
Note that \cite[Theorem 6]{glabbeek11ipl} can be reformulated as saying that a structural conflict
net $N$ that is not conflict-free fails to have a unique maximal BD-run. This implies that the set
of all partial BD-runs of $N$ fails to be a BD-run, and must hence fail to be directed.
This in turn implies that there are two finite BD-processes of $N$ without a common extension.

\section{A conflict-free structural conflict net has a largest BD-process}\label{sec-largestprocess}

In this section we prove the main result of this paper (\refthm{sclargest}),
namely that each conflict-free structural conflict net has a largest BD-process with respect to the
order $\sqsubseteq_1^\infty$. We make use of a labelled transition relation between the processes of
a given net. The fact that we are dealing with a structural conflict net is used only at the end of
the proof of \refthm{sclargest}.

Let $\PP = ((\SS, \TT, \FF, \MM_0), \pi)$ and $\PP' = ((\SS', \TT\,', \FF\,', \MM_0'), \pi')$
be GR-processes of a net $N = (S, T, F, M_0)$.
Henceforth, we will write $P'\goesto[a]P$ with $a\inp T$ a transition of the underlying net,
if $P'\mathbin\leq P$ and \plat{$\TT = \TT'\dcup\{t\}$} for some $t$ with $\pi(t)=a$.
Let $\mathcal{P}_0(N)$ be the set of \emph{initial processes} of a net $N$: those with an
empty set of transitions.
A process $P_0 \in \mathcal{P}_0(N)$ has exactly one place for each token in the
initial marking of $N$; two processes in $\mathcal{P}_0(N)$ differ only in the names of these places.
Now for each finite process $P$ of $N$, having $n$ transitions, there is a sequence
$P_0 \goesto[a_1] P_1 \goesto[a_2] \ldots \goesto[a_n] P_n$ with
$P_0\in\mathcal{P}_0(N)$ and $P_n=P$.

For $P=((\SS,\TT,\FF,\MM_0),\pi)$ a finite GR-process of a net $N=(S,T,F,M_0)$, we write
\plat{$\widehat P$} for the marking $\pi(P^\circ)\in \bbbn^S$.
The following observations describe a \emph{bisimulation}
between the above transition relation on the processes of a net, and the one on its markings.

\begin{observation}\rm\label{obs-bisimulation}
  Let $N=(S,T,F,M_0)$ be a net, $a\inp T$, and
  $P,Q\in\fGR$.\vspace{-1ex}
\begin{itemize}
\item[(a)]
$\mathcal{P}_0(N)\neq\emptyset$ and if $P\inp\mathcal{P}_0$ then $\widehat P = M_0$.
\item[(b)]
If $P \goesto[a] Q$ then $\widehat P \goesto[a] \widehat Q$.
\item[(c)]
If $\widehat P \goesto[a] M$ then there is a $Q$ with  $P \goesto[a] Q$ and $\widehat Q = M$.
\item[(d)]
$\widehat P$ is \emph{reachable} in the sense that $M_0 \goesto[] \widehat{P}$.
(This follows from (a) and (b).)
\end{itemize}
\end{observation}

\begin{lemma}\rm\label{lem-confl}
Let $P,P'\inp\fGR$, and $a,b$ transitions of the underlying net $N\!$.%
\vspace{1pt}

If $P\goesto[a]P'$ and $\widehat P\goesto[\{a,b\}]$ then
$\exists Q,Q'.~P'\goesto[b]Q' \wedge P \goesto[b] Q \goesto[a] Q'$.
\end{lemma}
\begin{proof}
Since $\widehat P\goesto[\{a,b\}]$ we have $\precond{a}\mathord+\precond{b}
\subseteq \pi(P^\circ)$.
Furthermore $\pi(P^\circ\setminus P'^\circ)=\precond{a}$.
So $\precond{b} \subseteq \pi(P^\circ \cap P'^\circ)$.
Therefore, there exist $Q$ and $Q'$ as required.
\qed
\end{proof}
The following observations are easy to check. For (b) note that $P\swapeq^*Q$ implies
  $\widehat P = \widehat Q$; also compare \refcor{swapprefix}.
\begin{observation}\rm\label{obs-swaptrans}
Let $P,Q,Q'$ be finite GR-processes of a net $N$.  \vspace{-1ex}
\begin{itemize}
\item[(a)] 
If $P\goesto[a]Q$ and $P\goesto[a]Q'$ then $Q\swapeq^* Q'$.
\item[(b)]
If $P\swapeq^*Q\goesto[a]Q'$ then $P\goesto[a]P'\swapeq^* Q'$ for some $P'\inp\fGR$.
\end{itemize}
\end{observation}

\begin{lemma}\rm\label{lem-cfdiamond}
  Let $N\mathbin=(S,T,F,M_0)$ be a binary-conflict-\!-free net,
  $a,b\inp T$ with $a\mathbin{\neq} b$, and $P,P',Q$ be finite GR-processes of $N$.

  If $P\goesto[a]P'$ and $P\goesto[b]Q$
  then $\widehat P \goesto[\{a,b\}]$ and $\exists Q'. P'\goesto[b]Q' \wedge Q\goesto[a]\swapeq^*Q'$.
\end{lemma}
\begin{proof}
  Suppose $P\goesto[a]P'$ and $P\goesto[b]Q$ with  $a\neq b$.
  We have $M_0\goesto[] \widehat P$ by \refobs{bisimulation}(d).
  Moreover, $\widehat P \goesto[a] \widehat P'$ and $\widehat P\goesto[b]\widehat Q$
  by \refobs{bisimulation}(b).
  Hence, as $N$ is binary-conflict-\!-free, $\widehat P \goesto[\{a,b\}]$.
  By \reflem{confl} there are $Q',Q''$ with \mbox{$P'\goesto[b]Q'$} and
  $P\goesto[b]Q''\goesto[a]Q'$.
  By \refobs{swaptrans}(a), $Q\swapeq^* Q''$, and hence $Q\goesto[a]\swapeq^* Q'$
  by \refobs{swaptrans}(b).
\qed
\end{proof}

\begin{lemma}[\cite{glabbeek11ipl}]\rm\label{lem-towardsinfinity}
  Let $N$ be a net.

  Every GR-process $\PP$ of $N$ is a prefix of a maximal GR-process of $N$.
\end{lemma}
\begin{proof}
  The set of all processes of $N$ of which $\PP$ is a prefix
  is partially ordered by $\leq$.
  Every chain in this set has an upper bound, obtained
  by componentwise union.
  Via Zorn's Lemma this set contains at least one maximal process.
  \qed
\end{proof}
Since the set of GR-processes of $N$ is non-empty by \refobs{bisimulation}(a),
this implies that each net has a maximal GR-process.

\begin{theorem}\rm\label{thm-sclargest}
  Let $N$ be a conflict-free structural conflict net.

  Then $N$ has a $\sqsubseteq_1^\infty$-largest BD-process.
\end{theorem}
\begin{proof}
  Let $P$ be a maximal GR-process of $N$---it exists by \reflem{towardsinfinity}.
  We show that $\BDinf{P}$ is the $\sqsubseteq_1^\infty$-largest BD-process of $N$,
  i.e., for each GR-process $Q$ of $N$ one has $Q \sqsubseteq_1^\infty P$.
  This proof is illustrated in \reffig{sclargest}.

  Let $\BDify{P}$ be as defined in \refsec{characterisation}. As remarked there, using \refcor{swapprefix},
  $\BDify{P} = \{P''' \mathbin\in \fGR \mid \exists P'',P'\mathbin\in\fGR. P''' \leq P'' \swapeqi^* P' \leq P\}$.

  Towards a contradiction, suppose $Q \not\sqsubseteq_1^\infty P$ for some $Q\in\GR$.
  Then, by \refdf{BD-swapping-alt}, there is a finite prefix $Q''$ of $Q$ with $Q''\notin\BDify{P}$.
  Let $Q_0$ be a minimal
  such prefix w.r.t.\ the prefix order $\leq$ of \refdf{extension}.
  $Q_0$ can be written as $((\SS, \TT, \FF, \MM_0), \pi)$.
  Since all initial processes of $N$ are isomorphic, each initial process of $N$ is in $\BDify{P}$.
  Hence $Q_0$ must have a transition.

\begin{figure}[H]
\definecolor{lsA}{rgb}{0.0,0.0,0.0}\def\lsA{[color=lsA,text=black,line width=1pt]}
\definecolor{lsB}{rgb}{0.4,0.0,0.4}\def\lsB{[color=lsB,text=black,line width=1pt]}
\definecolor{lsC}{rgb}{0.2,0.0,1.0}\def\lsC{[color=lsC,text=black,line width=1pt]}
\definecolor{lsD}{rgb}{0.0,0.2,0.8}\def\lsD{[color=lsD,text=black,line width=1pt]}
\definecolor{lsE}{rgb}{0.0,0.6,0.5}\def\lsE{[color=lsE,text=black,line width=1pt]}
\definecolor{lsF}{rgb}{0.0,0.8,0.2}\def\lsF{[color=lsF,text=black,line width=1pt]}
\definecolor{lsG}{rgb}{0.2,1.0,0.2}\def\lsG{[color=lsG,text=black,line width=1pt]}
\definecolor{lsH}{rgb}{0.6,0.8,0.0}\def\lsH{[color=lsH,text=black,line width=1pt]}
\definecolor{lsI}{rgb}{1.0,0.6,0.0}\def\lsI{[color=lsI,text=black,line width=1pt]}
\definecolor{lsJ}{rgb}{1.0,0.0,0.0}\def\lsJ{[color=lsJ,text=black,line width=1pt]}

\begin{center}
\begin{tikzpicture}
  \draw
    (6, -4.5) node(Q) {$Q$}
    (-2, -4.5) node(P) {$P$}
    (4, -18) node(Q0) {$Q_0$}
    (0, -18) node(Q0p) {$Q_0'$}
    (0, -12) node(P0p) {$P_0'$}
    (-2, -14) node(Qp) {$Q'$}
    (0, -16) node(Q1p) {$Q_1'$}
    (0, -14) node(Qnp) {$Q_n'$}
    (2, -19) node(t) {$+t$}
    (4, -16) node(Q1) {$Q_1$}
    (4, -14) node(Qn) {$Q_n$}
    (4, -12) node(P0) {$P_0$}
    (0, -10) node(P1p) {$P_1'$}
    (0, -6) node(Pm1p) {$P_{m+1}'$}
    (4, -10) node(P1) {$P_1$}
    (4, -6) node(Pm1) {$P_{m+1}$}
    (0, -8) node(Pmp) {$P_{m}'$}
    (4, -8) node(Pm) {$P_{m}$}
    ;

  \draw \lsA (Q) -- (P) node [sloped,midway,fill=white] {$\not\sqsupseteq_1^\infty$} ;
  \draw \lsB (Q0) to[out=0,in=270] (6, -16) -- (Q) node [sloped,midway,fill=white] {$\leq$} ;
  \draw \lsC (Q0p) to[out=330,in=180] (t) ;
  \draw \lsC (t) to[out=0,in=210] (Q0) ;
  \draw \lsD (Q0p) to[out=180,in=270] (-2, -16) -- (Qp) node [sloped,midway,fill=white] {$\leq$} ;
  \draw \lsD (Qp) -- (P0p) node [sloped,midway,fill=white] {$\swapeq^*$} ;
  \draw \lsD (P0p) to[out=180,in=270] (-2, -10) -- (P) node [sloped,near end,fill=white] {$\leq$} ;
  \draw \lsE (Q0p) -- (Q1p) node [sloped,midway,fill=white] {$\goesto[a_1]$} ;
  \draw \lsE (Qp) -- (Qnp) node [sloped,midway,fill=white] {$=$} ;
  \draw \lsE (Q1p) -- (Qnp) node [sloped,midway,fill=white] {$\cdots$} ;
  \draw \lsF (Q0p) -- (Q0) node [sloped,midway,fill=white] {$\goesto[b]$} ;
  \draw \lsF (t) -- (2, -18.2) node [sloped,midway,fill=white] {$\pi$} ;
  \draw \lsG (Q0) -- (Q1) node [sloped,midway,fill=white] {$\goesto[a_1]\swapeq^*$} ;
  \draw \lsG (Q1) -- (Qn) node [sloped,midway,fill=white] {$\cdots$} ;
  \draw \lsG (Q1p) -- (Q1) node [sloped,midway,fill=white] {$\goesto[b]$} ;
  \draw \lsG (Qnp) -- (Qn) node [sloped,midway,fill=white] {$\goesto[b]$} ;
  \draw \lsH (P0p) -- (P0) node [sloped,midway,fill=white] {$\goesto[b]$} ;
  \draw \lsH (Qn) -- (P0) node [sloped,midway,fill=white] {$\swapeq^*$} ;
  \draw \lsI (P0p) -- (P1p) node [sloped,midway,fill=white] {$\goesto[c_0]$} ;
  \draw \lsI (Pmp) -- (Pm1p) node [sloped,midway,fill=white] {$\goesto[c_m]$} ;
  \draw \lsI (Pm1p) -- (P) node [sloped,midway,fill=white] {$\geq$} ;
  \draw \lsI (P1p) -- (Pmp) node [sloped,midway,fill=white] {$\cdots$} ;
  \draw [->] \lsI (-2,-7) -- (-0.4,-7) node [sloped,midway,fill=white] {$\pi_P$} ;
  \draw \lsJ (P1p) -- (P1) node [sloped,midway,fill=white] {$\goesto[b]$} ;
  \draw \lsJ (Pmp) -- (Pm) node [sloped,midway,fill=white] {$\goesto[b]$} ;
  \draw \lsJ (Pm1p) -- (Pm1) node [sloped,midway,fill=white] {$\goesto[b]$} ;
  \draw \lsJ (P0) -- (P1) node [sloped,midway,fill=white] {$\goesto[c_0]\swapeq^*$} ;
  \draw \lsJ (P1) -- (Pm) node [sloped,midway,fill=white] {$\cdots$} ;
  \draw \lsJ (Pm) -- (Pm1) node [sloped,midway,fill=white] {$\goesto[c_m]\swapeq^*$} ;
\end{tikzpicture}
\end{center}
\caption{Illustration of the proof of \refthm{sclargest}.\label{fig-sclargest}}
\end{figure}

  Let $t$ be a maximal element in $\TT$ with respect to \plat{$\FF^+$}.
  Then $Q_0 \upharpoonright (\TT \setminus \{t\}) =: Q'_0$ is a process
  and $Q'_0 \mathbin{\in} \BDify{P}$.
  Hence there exists finite $P'_0, Q'$ such that $Q'_0 \leq Q' \swapeq^* P'_0 \leq P$.
  Moreover, there are $Q'_1,\ldots,Q'_n\in \BDify{P}$ and transitions $a_1\ldots,a_n$ of $N$
  with $Q'_n=Q'$ and $Q'_{i-1}\goesto[a_i]Q'_i$ for $i=1,\ldots,n$.

  $\pi(t)$ is some transition $b$ of $N$, so $Q'_0 \goesto[b] Q_0$.
  We now show by induction on $i\in\{1,\ldots,n\}$ that there are
  $Q_1,\ldots,Q_n\in \fGR\mathop{\setminus}\BDify{P} $
  with $Q'_i\goesto[b]Q_i$ and \mbox{$Q_{i-1}\goesto[a_i]\swapeq^* Q_i$} for $i\mathbin=1,\ldots,n$.
  Namely, given $Q_{i-1}$, as $Q_{i-1}\mathbin{\not\in}\BDify{P}$ we have
  $Q_{i-1}\not\swapeq^* Q'_i\in\BDify{P}$.
  Using that $Q'_{i-1} \goesto[a_i] Q'_i$ and $Q'_{i-1} \goesto[b] Q_{i-1}$,
  this implies $a_i\mathbin{\neq} b$ by \refobs{swaptrans}(a).
  Now \reflem{cfdiamond} yields a $Q_i\inp\fGR$ such that 
  $Q'_{i} \goesto[b] Q_{i}$ and $Q_{i-1} \goesto[a_i]\swapeq^* Q_i$.
  As $\BDify{P}$ is $\swapeq^*$- and prefix-closed, we have $Q_i\mathbin{\not\in}\BDify{P}$.

  Since $Q'_n\swapeq^* P'_0$ and $Q'_n\goesto[b] Q_{n}$,
  there is a $P_0\in\fGR$ with $P'_0\goesto[b] P_0$ and $P_0\swapeq^*Q_{n}$, using \refobs{swaptrans}(b).
  Hence $P_0\not\in\BDify{P}$.

  Now let $u$ be any transition in $P:=(\NN,\pi_P)$ that is not included in $P'_0$.
  Then there are $P'_1,\ldots,P'_{m+1} \leq P$ with
  $P'_{i}\goesto[c_i]P'_{i+1}$ for $i=0,\ldots,m$ and $c_m=\pi_P(u)$.
  Exactly as above, by induction on $i$, $b \mathbin{\neq} c_i$ for $i=0,\ldots,m$ and there are
  $P_1,\ldots,P_{m+1}\in \fGR\mathop{\setminus}\BDify{P}$\vspace{2pt}
  with $P'_{i+1}\goesto[b]P_{i+1}$ and \mbox{$P_{i}\goesto[c_i]\swapeq^{*} P_{i+1}$} for $i=0,\ldots,m$.
  Moreover, since \mbox{$P'_{m}\goesto[c_{m}]P'_{m+1}$} and $P'_{m}\goesto[b]$,
  we have \plat{$\widehat P'_m \goesto[\{c_m,b\}]$} by \reflem{cfdiamond}.
  By \refobs{bisimulation}(d) we furthermore have \plat{$M_0 \goesto[] \widehat
  P'_{m}$}, where $N=:(S,T,F,M_0)$.
  Hence, as $N$ is a structural conflict net, $\precond{b}\cap\precond{c_m}=\emptyset$.

  Since $\widehat P'_0 \supseteq \precond{b}$, by \refobs{bisimulation}(b), and the
  tokens in the preplaces of $b$ cannot be consumed by the $\pi_P$-image of any
  transition of $P$ that fires after $P'_0$ has been executed, $P$ can be extended
  with the transition $b$, and hence is not maximal.
  This is the required contradiction.
  \qed
\end{proof}

\section[Unique maximal GR-processes up to swapping equivalence]
        {Unique maximal GR-processes up to $\swapeq^\infty$}
\label{sec-maximality}

Together, Theorems~\ref{thm-conflictfree} and~\ref{thm-sclargest} say that
a structural conflict net has a $\sqsubseteq_1^\infty$-largest BD-process iff it is conflict-free.
The ``only if'' direction stems essentially from \cite{glabbeek11ipl}, and ``if'' is contributed here.

Since the preorder $\sqsubseteq_1^\infty$ was not employed in \cite{glabbeek11ipl},
there we did not consider $\sqsubseteq_1^\infty$-largest BD-processes.
Instead, we spoke of a ``unique maximal GR-process up to $\swapeq^\infty$'',
using the notion of maximality from \refdf{extension}, that is,
maximality w.r.t.\ the prefix order $\leq$ between GR-processes.
The following propositions compare $\sqsubseteq_1^\infty$-maximality and $\leq$-maximality.

\begin{proposition}\rm\label{pr-maxorder}
  Let $N$ be a net and $P$ a process thereof.

  If $\BDinf{P}$ is $\sqsubseteq_1^\infty$-maximal then some $Q\in\BDinf{P}$ is maximal.
\end{proposition}
\begin{proof}
  Assume $\BDinf{P}$ is $\sqsubseteq_1^\infty$-maximal.
  By \reflem{towardsinfinity} there exists some maximal $Q$ with $P \leq Q$.
  By \refdf{BD-swapping-alt}, $P \leq Q$ implies $P \sqsubseteq_1^\infty Q$.
  Since $\BDinf{P}$ is $\sqsubseteq_1^\infty$-maximal we have $Q \swapeq^\infty P$ and $Q$ is a
  maximal process within $\BDinf{P}$.
  \qed
\end{proof}

\begin{figure}
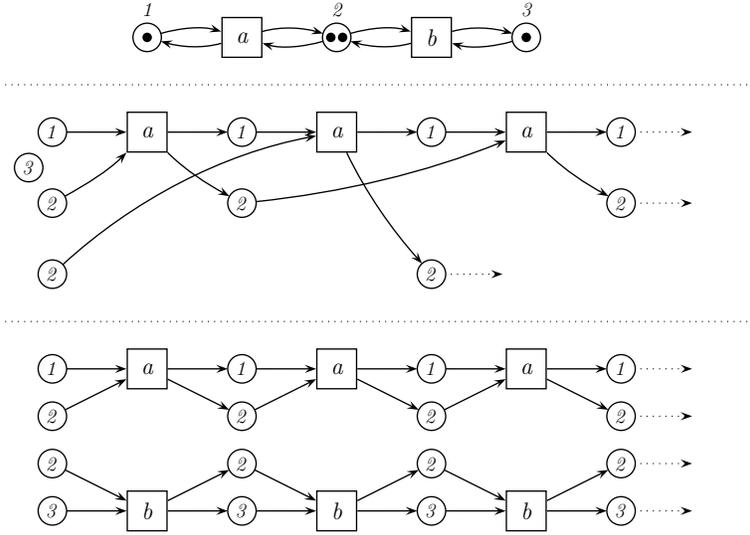

  \begin{center}
    \psscalebox{0.9}{
    \begin{petrinet}(16,11)
      \Qt(3,10):p:1;
      \t(5,10):a:a;
      \qt(7,10):q:2;
      \t(9,10):b:b;
      \Qt(11,10):r:3;
      \A p->a; \A a->p;
      \A q->a; \A a->q;
      \A q->b; \A b->q;
      \A r->b; \A b->r;
      \pscircle*(6.875,10){0.1}
      \pscircle*(7.125,10){0.1}

      \psline[linestyle=dotted](0,9)(16,9)

      \s(0.5,7.25):r1:3;
      \s(1,8):p1:1;
      \s(1,6.5):q1:2;
      \s(1,5):q2:2;

      \t(3,8):a1:a;
      \t(7,8):a2:a;
      \t(11,8):a3:a;

      \s(5,8):p2:1;
      \s(9,8):p3:1;
      \s(13,8):p4:1;

      \s(5,6.5):q3:2;
      \s(9,5):q4:2;
      \s(13,6.5):q5:2;

      \a p1->a1; \a a1->p2; \a p2->a2; \a a2->p3; \a p3->a3; \a a3->p4;
      \B q1->a1; \B a1->q3; \B q3->a3; \B a3->q5;
      \A q2->a2; \B a2->q4;

      \psline[linestyle=dotted]{->}(13.3,8)(14.5,8)
      \psline[linestyle=dotted]{->}(13.3,6.5)(14.5,6.5)
      \psline[linestyle=dotted]{->}(9.3,5)(10.5,5)

      \psline[linestyle=dotted](0,4)(16,4)

      \s(1,3):pb1:1;
      \s(1,2):qb1:2;
      \s(1,1):qb2:2;
      \s(1,0):rb1:3;

      \t(3,3):ab1:a;
      \t(7,3):ab2:a;
      \t(11,3):ab3:a;

      \t(3,0):bb1:b;
      \t(7,0):bb2:b;
      \t(11,0):bb3:b;

      \s(5,3):pb2:1;
      \s(9,3):pb3:1;
      \s(13,3):pb4:1;

      \s(5,2):qb3:2;
      \s(9,2):qb4:2;
      \s(13,2):qb5:2;

      \s(5,1):qb6:2;
      \s(9,1):qb7:2;
      \s(13,1):qb8:2;

      \s(5,0):rb2:3;
      \s(9,0):rb3:3;
      \s(13,0):rb4:3;

      \a pb1->ab1; \a ab1->pb2; \a pb2->ab2; \a ab2->pb3; \a pb3->ab3; \a ab3->pb4;
      \a qb1->ab1; \a ab1->qb3; \a qb3->ab2; \a ab2->qb4; \a qb4->ab3; \a ab3->qb5;
      \a qb2->bb1; \a bb1->qb6; \a qb6->bb2; \a bb2->qb7; \a qb7->bb3; \a bb3->qb8;
      \a rb1->bb1; \a bb1->rb2; \a rb2->bb2; \a bb2->rb3; \a rb3->bb3; \a bb3->rb4;

      \psline[linestyle=dotted]{->}(13.3,3)(14.5,3)
      \psline[linestyle=dotted]{->}(13.3,2)(14.5,2)
      \psline[linestyle=dotted]{->}(13.3,1)(14.5,1)
      \psline[linestyle=dotted]{->}(13.3,0)(14.5,0)

    \end{petrinet}
    }
  \end{center}
  \caption{A net and two maximal GR-processes thereof.}
  \label{fig-weakmaximal}
\vspace{5ex}
\end{figure}

\noindent
The reverse of \refpr{maxorder} does not hold.
The first process depicted in \reffig{weakmaximal} cannot be extended, for none of the tokens
in place 2 will in the end come to rest. So it is maximal.
Yet, it is not $\sqsubseteq_1^\infty$-maximal.
For it is swapping equivalent with the top half of the second process (using only one
of the tokens in place 2), which can be extended with the bottom half.

\begin{proposition}\rm\label{pr-unique-max}
  Let $N$ be a net and $P$ a process thereof.

  If $P$ is the only maximal process up to $\swapeq^\infty$, then it is the
  $\sqsubseteq_1^\infty$-largest process.
\end{proposition}
\begin{proof}
Let $P$ be the only maximal process of $N$ up to $\swapeq^\infty$, and $Q$ any other process of $N$.
Let $Q'$ be a maximal process with $Q\leq Q'$---it exists by \reflem{towardsinfinity}.
Using \refdf{BD-swapping-alt}, $Q\leq Q'$ trivially implies $Q\sqsubseteq_1^\infty Q'$.
Since $P$ is the only maximal process up to $\swapeq^\infty$, we have $Q' \swapeq^\infty P$.
Thus $Q \sqsubseteq_1^\infty P$, showing that $P$ is the $\sqsubseteq_1^\infty$-largest process of $N$.
  \qed
\end{proof}

\cite[Corollary 1]{glabbeek11ipl} says that if a structural conflict net $N$ has only one maximal
GR-process up to $\swapeq^\infty$ then $N$ is conflict-free. Using \refpr{unique-max} this is a
weakening of \refthm{conflictfree}. We now establish the converse, that a conflict-free structural
conflict net has only one maximal GR-process up to $\swapeq^\infty$; this is a strengthening of \refthm{sclargest}.

\begin{theorem}\rm\label{thm-scmaximals}
  Let $N$ be a conflict-free structural conflict net.

  Then $N$ has a unique maximal GR-process up to $\swapeq^\infty$.
\end{theorem}
\begin{proof}
  Let $P$ and $Q$ be two maximal GR-processes of $N$.
  The proof of \refthm{sclargest} shows that $\BDinf{P}$ is the $\sqsubseteq_1^\infty$-largest BD-process of $N$,
  and the same holds for $Q$. So $Q \sqsubseteq_1^\infty P$ and $P \sqsubseteq_1^\infty Q$, i.e.,
  $Q \swapeq^\infty P$.
\qed
\end{proof}
Thus we obtained, for structural conflict nets $N$, that
$N$ is conflict-free iff
$N$ has a $\sqsubseteq_1^\infty$-largest BD-process,
iff $N$ has a unique maximal GR-process up to $\swapeq^\infty$.
In our technical report \cite{GGS11a} we moreover show that for structural conflict nets the
converse of \refpr{maxorder} holds: if $P\in\GR$ is maximal, then $\BDinf{P}$ is
$\sqsubseteq_1^\infty$-maximal---see Lemma~12. Consequently, for structural conflict nets also the
converse of \refpr{unique-max} holds. So for structural conflict nets there is no difference between
a $\sqsubseteq_1^\infty$-largest BD-process and a unique maximal GR-process up to $\swapeq^\infty$.

\section{Conclusion}
We defined a \emph{BD-process} as an equivalence class of Goltz-Reisig
processes under the swapping equivalence proposed by Best and
Devillers, and argued that on the subclass of structural conflict nets
BD-processes constitute a fully satisfactory concept of abstract
process of a Petri net under the collective token interpretation.  To
justify that assessment we showed that a structural conflict net is
conflict-free iff it has a largest BD-process.

In the technical report belonging to \cite{GGS11a} we strengthen the result obtained here
by showing that each countable net without binary conflict (even if not a structural conflict net)
has a largest BD-process. However, proving this is much more complicated than the results presented here.

We leave as an open question to consider also branching time semantics.
The notion of a process for condition/event systems was adapted to a
branching time semantics of nets through the concept of an
\emph{unfolding} of Nielsen, Plotkin {\AND} Winskel \cite{nielsen81petri}. 
Unfolding a net results in an occurrence net with forward branched places
that captures all runs of the net, together with the branching
structure of choices between them.
This work was adapted by Engelfriet in \cite{engelfriet91branchingprocesses}
to P\hspace{-1pt}/T systems without arc weights, and
Meseguer, Sassone {\AND} Montanari extended this to cover arc weights as
well \cite{MMS97}. The resulting occurrence nets have one branch for
every maximal GR-process of the underlying net.
It is an open question whether such a construction can be adapted to
the collective token interpretation of Petri nets, so that an
unfolding of a net has one branch for every BD-process, and thus
remains unbranched in case of conflict-free nets.

\bibliographystyle{eptcsalpha}

\end{document}